\documentclass[12pt]{amsart}
\usepackage[utf8]{inputenc}
\usepackage{pdfpages}
\usepackage{mathptmx}
\usepackage{amssymb}
\usepackage{mathtools}
\usepackage{lipsum}
\usepackage{tikz}
\usetikzlibrary{cd}
\usepackage{verbatim}
\usepackage{pdfpages}
\usepackage[
backend=biber,
style=alphabetic,
sorting=nty,
maxbibnames=99
]{biblatex}
\usepackage{combelow}
\usepackage{enumitem}
\usepackage[colorlinks=true,linkcolor=blue,citecolor=red]{hyperref}
\usetikzlibrary{matrix}
\usepackage{braket}
\usepackage{physics}
\usepackage{algorithm}
\usepackage{algpseudocode}
\usepackage[left=30mm, top=20mm, right=30mm, bottom=25mm]{geometry}
\long\def\comment#1\endcomment{}

\newcommand{\bA}{\textbf{A}}

\newcommand{\MIP}{\text{MIP}}
\newcommand{\LIN}{\text{LIN-MIP}}
\newcommand{\RE}{\text{RE}}

\newcommand{\halt}{\text{Halt}}
\newcommand{\poly}{\text{poly}}

\newtheorem{prop}{Proposition}[section]
\newtheorem{thm}[prop]{Theorem}
\newtheorem{lem}[prop]{Lemma}
\newtheorem{remark}[prop]{Remark}

\newcommand{\bbZ}{{\mathbb Z}}
\newcommand{\bbI}{{\mathbb I}}
\newcommand{\bbR}{{\mathbb R}}
\newcommand{\bbN}{{\mathbb N}}
\newcommand{\bbC}{{\mathbb C}}

\newcommand{\bbE}{{\mathbb E}}

\newcommand{\calA}{{\mathcal A}}
\newcommand{\calH}{{\mathcal H}}
\newcommand{\calG}{{\mathcal G}}
\newcommand{\calF}{{\mathcal F}}
\newcommand{\calI}{{\mathcal I}}
\newcommand{\calO}{{\mathcal O}}

\newcommand{\lin}{\text{LIN}}
\newcommand{\supp}{\text{supp}}

\newtheorem{theorem}{Theorem}[section]
\newtheorem{claim}[theorem]{Claim}
\newtheorem{corollary}[theorem]{Corollary}

\newtheorem{definition}[theorem]{Definition}

\theoremstyle{remark}

\title{Approximating the quantum value of an LCS game is RE-hard}
\author{Aviv Taller and Thomas Vidick}
\address{}
\email{}

\begin{document}

\maketitle

\begin{abstract}
     We generalize H\r{a}stad's long-code test for projection games and show that it remains complete and sound against entangled provers.
     Combined with a result of Dong et al. \cite{Dong25}, which establishes that $\MIP^*=\RE$ with constant-length answers, we derive that $\LIN^*_{1-\epsilon,s}=\RE$, for some $1/2< s<1$ and for every sufficiently small $\epsilon>0$, where LIN refers to linearity (over $\mathbb{F}_2$) of the verifier predicate. Achieving the same result with $\epsilon=0$ would imply the existence of a non-hyperlinear group. 
     
\end{abstract}

\section{Introduction}

In a seminal result \cite{Hstad2001}, H\r{a}stad established that the problem of approximating the maximum winning probability of two classical provers in the nonlocal game associated to a {\em linear constraint system} (LCS) is NP-hard.
In this nonlocal game, or two-prover one-round interactive protocol, the verifier sends to one prover (Alice) a linear equation out of the system and one of the variables in that equation to the other prover (Bob).
Unaware of the question posed to the other prover, each participant is required to provide an assignment for their respective variables.
The provers, or players, win if Alice's assignment satisfies the equation and agrees with Bob's assignment for the selected variable.
H\r{a}stad's proof of hardness of approximation for the maximum winning probability in an LCS game, referred to as its {\em classical game value}, consists of three main components.
The first component involves the association of a 3SAT formula with a long-code test, as constructed in \cite{Hstad2001}, which has been demonstrated to possess (imperfect) completeness and soundness.
At the heart of the long-code test lies a distorted linear test on three bits.
This results in the formulation of a linear system, and the subsequent hardness result is derived through its combination with the other two components, specifically the PCP theorem \cite{Arora1998,Arora1998b} and Raz's parallel repetition theorem \cite{Raz98}.
It is noted that the preceding discussion pertains to classical provers; specifically, these are two randomized algorithms which are permitted to share a mutual random seed.

In this work, we establish that approximating the maximum winning probability of {\em quantum provers} in an LCS game, referred to as its {\em quantum game value}, is RE-hard.
By quantum provers, we mean provers who are allowed to perform measurements on a finite dimensional shared entangled state.
In other words, we show that $\LIN^*_{1-\epsilon,s}=\RE$ for certain $1>s> 1/2$ and for every sufficiently small $\epsilon>0$.

In order to obtain this hardness result, we needed quantum analogs for the three classical components of the proof.
The analog for the first component is provided in this work.
Namely, we show that the long-code test is complete and sound against quantum provers.
This is our main technical contribution.
In place of the PCP theorem, we employ the equivalence of the classes $\MIP^*=\RE$, as formulated by Dong et al. \cite{Dong25}.
In that work, the authors refine the reduction of the halting problem as established in \cite{Ji2021}, guaranteeing that the resulting $\MIP^*$ protocols use polynomial length questions while maintaining answers of constant length.
Finally, in lieu of the classical parallel repetition theorem, we utilize the parallel repetition theorem for entangled projection games by Dinur et al. \cite{Dinur2015}.
Loosely speaking, projection games refer to games in which the response of Bob, which leads to a win, is uniquely determined by Alice’s response.

Furthermore, a fourth component that we needed and which allowed us to combine the other three is an observation provided by \cite{CM25}.
This observation states, roughly, that the reduction of \cite{Dong25} gives rise to a reduction of the halting problem to projection {\em Boolean constraint system} games with equations of constant length.
Informally, a Boolean constraint system (BCS) is defined as any system with binary equations.
In the aforementioned reduction, YES instances are mapped to BCSs admitting a special type of perfect quantum strategies, while NO instances are sent to BCSs with quantum value bounded away from 1.

We discuss the soundness parameter that is obtained through our reduction. 
The linear system induced by the test consists exclusively of equations that involve exactly three variables. 
Since any such system admits an assignment that satisfies at least half of the equations, this implies a lower bound of $(2+1/2)/3=5/6$ on the classical winning probability in the corresponding LCS game. 
Indeed, Bob may respond using such an assignment, and in at most half of the equations, Alice needs to modify at most one variable to ensure satisfiability. 
Following the analysis in \cite{Hstad2001}, this lower bound is in fact tight for classical provers. 
In contrast, the soundness bound we obtain for entangled provers is $35/36$, representing a square loss relative to the classical threshold. We do not know whether this bound is tight. 
Determining the minimal winning probability for entangled strategies in such LCS games, and whether it coincides with or exceeds the classical value, is left for future work.

The following is a brief discussion concerning the imperfect completeness of the main result.
This property is exhibited by both the original reduction from 3SAT in \cite{Hstad2001} and our reduction from the halting problem.
In both cases, this is a direct consequence of introducing noise into the linearity test, which is a necessary component of the construction. 
However, beyond this, there are additional barriers to achieving perfect completeness using our methods.
From the computational point of view, an efficient reduction with perfect completeness for classical strategies would imply $\text{P}=\text{NP}$, which is widely believed to be false.
On top of that, there are also algebraic obstacles, which we now describe in more detail.

Each BCS game $B$ can be associated with a $*$-algebra, denoted $\calA(B)$.
Various notions of satisfiability for the game $B$, that is, different types of perfect strategies with winning probability 1, correspond to different types of $*$-representations of $\calA(B)$.
It was shown in \cite{PaddockSlofstra2023} that for LCS games, certain types of satisfiability are equivalent at the algebraic level, whereas this equivalence fails for general BCS games.
In particular, it is not possible, in general, to associate with every BCS $B$ an LCS $B'$ such that $\calA(B)$ admits a $*$-morphism into $\calA(B')$.
This fact presents a significant obstacle to reducing the known $\MIP^*$ protocols for the halting problem to LCS games with perfect completeness. 
Indeed, these protocols admit a presentation as BCS protocols.
As a result, any attempt to generically reduce such protocols to LCS protocols, in a way that induces a $*$-embedding of the associated algebras, is ruled out by the aforementioned algebraic constraint.
Further details on reductions of the halting problem to BCS protocols can be found in \cite{MastelSlofstra24,CM25,Fritz2020,PaddockSlofstra2023}.

One main motivation for pursuing perfect completeness lies in its connection to the longstanding open question concerning the existence of non-hyperlinear groups.
Roughly speaking, these groups admit no asymptotically faithful unitary action on finite-dimensional Hilbert spaces.
It is conjectured that such objects, if they exit, exhibit some exotic properties.
One way of proving their existence is by showing that $\LIN^*_{1,s}=\RE$, for some constant $1>s\geq 0$ \cite{Slofstra2019}.
For more details on non-hyperlinear groups and related topics, see \cite{Pestov2008}.

This section concludes with a review of some related work.
In \cite{CM25}, Culf and Mastel showed that for a list of NP-complete gap-problems, the associated entangled nonlocal game becomes $\MIP^*$-complete.
Their basic premise was that in the context of BCS games, synchronous strategies, which translate into approximate representations of the associated $*$-algebra, are the correct analog to a classical assignment for the constraint system.
Using this idea, that originates in \cite{MastelSlofstra24}, it was then proved that many classical reductions are still complete and sound against quantum strategies. 
In our work, we also use the same analogy, but only at the level of the strategies.
This key difference is a necessity.
Indeed, the reductions in \cite{CM25} were achieved by providing $*$-morphisms between the associated algebras, a method that in general cannot work with LCSs, as discussed above.

Another related work is of Man\u{c}inska, Spaas and Spirig \cite{MSS25}, in which it was proved that the gap-problem related to the class of independent set games is $\MIP^*$-complete.
Roughly speaking, an independent set game is a nonlocal game, in which the provers try to convince the verifier that a certain graph admits an independent set of a certain size.
We note that classically, these games are easy, that is, in P.
In that work, the authors showed that the reduction introduced in \cite{Maninska2015} is in fact sound against entangled provers.
To achieve this, they first proved a stability theorem for tracial von Neuman algebras, and then applied it on the tracial von Neuman algebra associated to an independent set game with a given synchronous strategy.
This kind of association was established for synchronous games \cite{Paulsen2016}, and to the best of our knowledge does not hold in general.
Consequently, this technique seems less suitable for cases in which the range of the reduction consists of non-synchronous games, as in our case.

We note that an earlier work by the second author~\cite{vidick2016three} already provides a reduction from \emph{three-prover} protocols with constant answer size to \emph{three-prover} protocols with linear (i.e., XOR) decision predicate, which parallels H\r{a}stad's reduction. 
The main difference with the present work is that the earlier paper requires three provers, whereas in this work we are able to carry out the reduction with two provers.
Working with only two provers adds additional complications, but it is necessary for our hardness result.
We discuss this issue in more detail at the end of Section \ref{sec:test}.
We note that~\cite{vidick2016three} was withdrawn due to an error in the underlying NP-hardness result, which does not affect the validity of the aforementioned reduction. 

Finally, after we obtained our results it was brought to our attention that a similar result was obtained independently in an unpublished work by O'Donnell and Yuen \cite{ODY}.

The remainder of the paper is organized as follows. Section \ref{sec:Preliminaries} covers preliminary concepts and known results in classical and quantum complexity. Section \ref{sec:BCS} defines Boolean and linear constraint system games and briefly describes some of their properties. Section \ref{sec:test} describes the construction of the entangled long-code test for BCS games. Section \ref{sec:main} defines the classes $\MIP^*$ and $\LIN^*$ and contains the proof of the main theorem.
Finally, Section \ref{sec:general_q.strategies} sketches how to generalize the main result from our proof for synchronous quantum strategies to general quantum strategies (cf. remark \ref{rem:quantum_vs_synch}).


\section{Preliminaries}\label{sec:Preliminaries}

\subsection{Boolean functions}
For any number $k\in\bbN$, we denote by $[k]$ the set $\{1,2,..,k\}$.
We identify subsets $\alpha\subset[k]$ with elements in $\{0,1\}^k$ in a natural way.
Given two elements $\alpha,\beta\in \{0,1\}^k$ we use the notation $\alpha\Delta\beta$ to denote their symmetric difference.
In some parts of this work we will use the multiplicative presentation of binary operations, where we identify $1$ with $0$ and 'false', and $-1$ with $1$ and 'true'.
We sometimes use the notation $\bbZ_2=\{\pm1\}$.
Given two elements of any kind $f$ and $g$, the function $\delta_{f,g}$ equals 1 if $f=g$ and 0 otherwise.

Given two finite sets $U\subset W$ and an element $x\in\{\pm1\}^W$, we denote its restriction to $U$ by $x|_U$.
Given a subset $\alpha\subset\{\pm1\}^W$, we denote by $\pi_2^U(\alpha)\subset \{\pm1\}^U$ the subset consisting of all elements $x\in\{\pm1\}^U$ such that there is an odd number of elements $y\in\alpha$ with $y|_U=x$.
Denote by $\calF_U$ the collection of functions $f:\{\pm1\}^U\rightarrow \{\pm1\}$.
Given $f\in\calF_U$, define $m(f)\in\{\pm1\}$ as follows
\begin{align*}
    m(f)=\begin{cases}
        -1 & \quad |f^{-1}(-1)|>|f^{-1}(1)| \\
        1 & \quad o/w
    \end{cases}
\end{align*}

For a set $\alpha\subset \{\pm1\}^U$, we define the function $\chi_{\alpha}:\calF_U\rightarrow \{\pm1\}$ by $\chi_{\alpha}(f)=\prod_{x\in\alpha}f(x)$.
In particular, $\chi_{\alpha}$ is linear, and for every $f\in\calF_U$, it holds that $\bbE_{\alpha}[\chi_{\alpha}(f)]=\delta_{\overline{1},f}$,
where $\overline{1}\in\{\pm1\}^U$ is the constant element with value $1$.
We define the {\em section function} $s_U:\calF_U\rightarrow \calF_U$, by $s_U(f)=f$ if $f(\overline{1})=1$ and $s_U(f)=-f$ otherwise.
It is a section in the sense that $s_U(f)=s_U(-f)\in\{f,-f\}$.
In particular, note that $m(f\cdot s_U(f))=-m(-f\cdot s_U(-f))$.

\subsection{Boolean matrices}
Let $U$ be a finite set.
The answers to the long-code test queries are interpreted as evaluations of some function $A_U:\calF_U\rightarrow \{\pm1\}$. 
Its goal is to validate the existence of an element $x\in\{\pm1\}^U$ such that $A_U(f)=f(x)$ for every $f\in\calF_U$.
In our generalized setting, each query with binary answer is associated with a {\em binary observable}.
A binary observable is a unitary involution $A\in U(\calH)$, where $\calH$ is any Hilbert space.
That is, $A^*=A$ and $A^2=\bbI_{\calH}$, where $A^*=A^{\dagger}$ is the conjugate transpose of $A$.
Following the foregoing discussion, we fix a collection $\bA:=\{A_f\}_{f\in \calF_U}$ of binary observables, all acting on the same space $\calH$, for the rest of this section.

\begin{definition}[Fourier Transform]\label{def:fourier}
    The Fourier coefficient of $\bA$ at $\alpha\subset\{\pm1\}^U$ is defined by 
    $$\hat{A}_{\alpha}:=(\bA,\chi_{\alpha}):=\underset{f\in\calF_U}{\bbE}[\chi_{\alpha}(f)A_f]\;.$$
\end{definition}

\begin{lem}[Fourier inversion formula and Parseval's identity]\label{lem:fourierprop}
    The following classical identities still hold: \begin{enumerate}
        \item Fourier inversion formula
        $$A_f:=\sum_{\alpha\subset\{\pm1\}^U}\chi_{\alpha}(f)\hat{A}_{\alpha}\;.$$
        \item Parseval's identity
        $$\sum_{\alpha}\hat{A}_{\alpha}^2=\bbI\;.$$
    \end{enumerate}
\end{lem}

\begin{proof}
    We begin with the inversion formula:
    \begin{align*}
        \sum_{\alpha\subset\{\pm1\}^U}\chi_{\alpha}(f)\hat{A}_{\alpha}&= \sum_{\alpha\subset\{\pm1\}^U}\chi_{\alpha}(f)\underset{f'\in\calF_U}{\bbE}[\chi_{\alpha}(f')A_{f'}]\\
        &= \underset{f'\in\calF_U}{\bbE}[A_{f'}\sum_{\alpha\subset\{\pm1\}^U}\chi_{\alpha}(f+f')]\\
        &=\underset{f'\in\calF_U}{\bbE}[A_{f'}2^{2^{|U|}}\delta_{f,f'}]=A_f\;.
    \end{align*}

    In addition,
    \begin{align*}
        \sum_{\alpha}\hat{A}_{\alpha}^2 &=\underset{f\in\calF_U}{\bbE}[A_f\sum_{\alpha}\chi_{\alpha}(f)\hat{A}_{\alpha}]\\
        &=\underset{f\in\calF_U}{\bbE}[A_f^2]\\
        &=\underset{f\in\calF_U}{\bbE}[\bbI]=\bbI\;.
    \end{align*}
\end{proof}

\begin{definition}[Folding over true]\label{def:folding}
    Given a collection of binary observables $\bA$ as above and $f\in\calF_U$, we define $A_{true,f}:=m(f\cdot s_U(f))A_{s_U(f)}$.
\end{definition}

\begin{lem}\label{lem:folding}
   The following holds:
   \begin{enumerate}
       \item $A_{true,f}=-A_{true,-f}$
       \item If $\alpha\subset \{\pm1\}^U$ is such that $|\alpha|$ is even, then $\hat{A}_{true,\alpha}=0$
   \end{enumerate}
\end{lem}
\begin{proof}
    The first assertion is immediate.
    For the second part, note that it is implied by the fact that for $|\alpha|$ even, $\chi_{\alpha}(f)=\chi_{\alpha}(-f)$.
    Thus, together with the first part and the definition of $\hat{A}_{true,\alpha}$, the assertion follows.
\end{proof}

\begin{definition}[Conditioning upon a function]\label{def:conditioning}
    Given $\bA$ as above and $C\in\calF_U$, we define for $f\in\calF_U$, $A_{C,f}:=A_{f\wedge C}$.
\end{definition}

\begin{lem}\label{lem:conditioning}
    Let $\alpha\subset \{\pm1\}^U$ be such that there exists $x\in \alpha$ with $C(x)=1$,\footnote{Recall that $1$ is associated with the Boolean value \emph{false}.} then $\hat{A}_{C,\alpha}=0$.
\end{lem}

\begin{proof}
    Let $x_0\in \alpha$ be such that $C(x_0)=1$.
    For $f\in\calF_U$, define $f'$ by $f'(x)=-f(x)$ for $x=x_0$ and $f'(x)=f(x)$ otherwise.
    Note that $f\wedge C=f'\wedge C$, and therefore, $A_{C,f}=A_{C,f'}$.
    On the other hand, $\chi_{\alpha}(f)=-\chi_{\alpha}(f')$. 
    The claim follows from the definition of $\hat{A}_{C,\alpha}$.
\end{proof}

Given $C\in\calF_U$, which is not the constant $1$ function, that is, the empty set, we can fold over true and condition upon $C$ simultaneously.
Given a pair of the form $g\wedge C$ and $(-g)\wedge C$, we choose one of them, denoting it by $s_{g,C}$.
Define $m_{f,C}=1$ if $s_{f,C}=f\wedge C$ and $-1$ otherwise.
Then we define $A_{true,C,f}=m_{f,c}A_{s_{f,C}}$.
Note that the definition of $A_{true,C,f}$ only depends on $f\wedge C$ and $A_{true,C,f}=-A_{true,C,-f}$.

\subsection{Measurements and general results in linear algebra}
A {\em positive operator-valued measure} (POVM) on a Hilbert space $\calH$, is a collection of positive semi-definite operators $\{A_{\alpha}\}_{\alpha\in I}$ on $\calH$ such that $\sum_{\alpha}A_{\alpha}=\bbI_{\calH}$.
A {\em projection-valued measure} (PVM) is a POVM $\{A_{\alpha}\}_{\alpha\in I}$, in which $A_\alpha$ is an orthogonal projection for every $\alpha$. 
A {\em state} is a unit vector in the Hilbert space $\calH$.
We use the bra-ket notation for vectors; thus, while $\ket{\psi}\in\calH$ denotes a vector, $\bra{\psi}=(\ket{\psi})^{\dagger}$ denotes its Hermitian adjoint.

Recall the {\em normalized Hilbert-Schmidt} inner product of two complex square matrices of dimension $d$, $\langle A,B\rangle:=\frac{1}{d}\Tr(A^{\dagger}B)$.
We denote the corresponding norm by $\|A\|_{hs}:=\sqrt{\langle A,A\rangle}$.

\begin{lem}\label{lem:general1}
    Let $Y_1,Y_2,Y_3,X_1,X_2$ and $X_3$ be three binary observables in $U(d)$.
    Then,
    $$|\frac{1}{d}\Tr( Y_1Y_2Y_3-X_1X_2X_3)|\leq (6(3-\sum_{l}\frac{1}{d}\Tr(Y_lX_l)))^{\frac{1}{2}}\;.$$
\end{lem}

\begin{proof}
    We note that for a binary observable $U$, $U^{\dagger}=U$. 
    Thus,
    \begin{align*}
        (*)\ |\frac{1}{d}\Tr( Y_1Y_2Y_3-X_1X_2X_3)|&\leq |\frac{1}{d}\Tr((Y_1-X_1)Y_2Y_3)|+|\frac{1}{d}\Tr(X_1(Y_2-X_2)Y_3)|\\
        &+|\frac{1}{d}\Tr(X_1X_2(Y_3-X_3))|\\ &\underset{CS+trace}{\leq}\|Y_1-X_1\|_{hs}\|Y_2Y_3\|_{hs}+\|Y_2-X_2\|_{hs}\|Y_3X_1\|_{hs}\\
        &\ \ \ \ \ \ + \|Y_3-X_3\|_{hs}\|X_1X_2\|_{hs}\;.
    \end{align*}
    Now, note that for an unitary $U$, $\|U\|_{hs}=1$, and for two binary observables $X $ and $Y$ we have $\|Y-X\|_{hs}^2=2(1-\frac{1}{d}\Tr(YX))$.
    Therefore,
    \begin{align*}
        (*) &\leq \sqrt{2}(\sum_l(1-\frac{1}{d}\Tr(Y_lX_l))^{\frac{1}{2}})\;.
    \end{align*}
    Finally, in $\bbC^3$, we have $\|x\|_1\leq \sqrt{3}\|x\|_2$, thus
    $$(*)\leq (6\sum_l1-\frac{1}{d}\Tr(Y_lX_l))^{\frac{1}{2}}\;,$$
    as needed.
\end{proof}

We recall the following version of the Cauchy–Schwarz inequality.
\begin{lem}[Cauchy–Schwarz inequality]\label{lem:CS}
    Let $d\in \bbN$ and let $\{A_{\beta}\}_{\beta\in I}$ and $\{B_{\beta}\}_{\beta\in I}$ be two collections of matrices in $M_d(\bbC)$.
    Then,
    $$|\sum_{\beta}\langle A_{\beta},B_{\beta}\rangle|\leq (\sum_{\beta}\langle A_{\beta},A_{\beta}\rangle)^{\frac{1}{2}}(\sum_{\beta}\langle B_{\beta},B_{\beta}\rangle)^{\frac{1}{2}}\;.$$
\end{lem}

\subsection{Nonlocal games}\label{sec:nonlocalgames}
A {\em nonlocal game} is a tuple $\calG = (\calI,\calO,\pi,D)$, where $\calI$ and $\calO$ are finite sets, $\pi$ is a probability distribution on the set $\calI^2$ and $D$ is a predicate $\calI^2\times\calO^2\rightarrow\{0,1\}$.
The game $\calG$ is called synchronous if in addition $\pi$ is symmetric and $D$ satisfies the synchronous condition $D(x,x,a,b)=0$ for every $x\in \calI$ and $a\neq b\in\calO$.

We also allow replacing $\calO$ by a collection of finite sets $\{\calO_i\}_{i\in\calI}$.
In this case, we add the condition $D(i_1,i_2,a_1,a_2)=0$ whenever $(a_1,a_2)\notin\calO_{i_1}\times \calO_{i_2}$.

The above definition can be interpreted as a two-player, one-round game with a referee.
In this game, the referee samples a pair of questions $x,y\sim \pi$ and sends each question to a different player (traditionally called Alice and Bob).
The players then need to respond with a pair of answers $(a,b)\in \calO$.
They do so without communicating with each other (and in particular without knowing what the other player received).
Finally, the referee evaluates $D(x,y,a,b)$ and declares that the players have won if and only if it equals $1$.

To improve their chances of winning, the players can agree in advance on a shared strategy.
A {\em quantum strategy} consists of two finite-dimensional Hilbert spaces $\calH_A$ and $\calH_B$, a state $\ket{\psi}\in \calH_A\otimes\calH_B$ and collections of POVM's $\{A^x_a\}_{a\in\calO}$ and $\{B^x_a\}_{b\in\calO}$ for every $x\in\calI$.
A quantum strategy is called {\em synchronous} if $\calH=\calH_A=\calH_B$, $\ket{\psi}$ is the maximally entangled state $\ket{\psi_{MES}}:=\frac{1}{\sqrt{d}}\sum_{i\in[d]}\ket{i}\ket{i}$ for some orthonormal basis of $\calH$, $\{A^x_a\}_a$ is a PVM and $(A^x_a)^T=B^x_a$ for every $x\in\calI$ and $a\in\calO$.
A synchronous strategy is called {\em oracularizable} if $A_a^xA^y_b=A^y_bA^x_a$ for every $(x,y)\in\supp(\pi)$ and $a,b\in\calO $, where for a probability measure $\mu$ on a finite set $U$, we define $\supp(\mu)$ to be the subset of elements $x\in U$ such that $\mu(x)>0$.

Every strategy gives rise to a {\em quantum correlation}, that is, a collection of probability distributions that takes the following form
$$(*)\ p(a,b|x,y)=\bra{\psi}A^x_a\otimes B^y_b\ket{\psi}\;.$$

If $\ket{\psi}=\ket{\psi_{MES}}$, then we have the identity $\bra{\psi}A\otimes B\ket{\psi}=\frac{1}{d}\Tr(AB^T)$.
So, synchronous correlations can be further presented as
$$ p(a,b|x,y)=\frac{1}{d}\Tr(A^x_aA^y_b)=\langle A^x_a,A^y_b\rangle$$
for every $x,y\in \calI$ and $a,b\in\calO$.
We will refer to strategies as correlations and vice versa, and treat them as synonyms.

The interpretation of a strategy is as follows.
Each of the players has their own system $\calH_A$ and $\calH_B$, and before the game starts, they prepare the state $\ket{\psi}$ in their combined system.
In the game, upon receiving the question $x$, Alice measures her system using the POVM $\{A_a^x\}_{a\in\calO}$ and responds according to the measurement result. 
Upon receiving the question $y$, Bob will act similarly with his measurements.
According to the laws of quantum mechanics, the probability that Alice's answer will be $a$ and Bob's answer will be $b$ is exactly given by $(*)$.

The {\em winning probability} of a correlation $p$ is defined by
$$\omega(\calG,p):=\underset{x,y\sim \pi}{\bbE}\bigg[\sum_{a,b\in\calO}p(a,b|x,y)D(x,y,a,b)\bigg].$$

The {\em bias} of $p$ in the game $\calG$ is defined as $\beta(\calG,p)=2\omega(\calG,p)-1$.
A strategy $p$ is called {\em perfect} if $\omega(\calG,p)=1$.
The {\em (quantum) game value} of $\calG$ is then defined as 
$$\omega_q(\calG)=\sup_p\omega(\calG,p)\;,$$
where the supremum is over all possible quantum strategies.
Finally, we define the {\em bias} of $\calG$ as $\beta_q(\calG)=2\omega_q(\calG)-1$.
The {\em synchronous (quantum) game value} of $\calG$, denoted by $\omega_q^s(\calG)$, is almost the same, except that the supremum is taken over all possible synchronous correlations.
The synchronous bias $\beta_q^s(\calG)$ is defined analogously.
Clearly, for every nonlocal game $\omega_q^s(\calG)\leq \omega_q(\calG)$.

The main example of nonlocal games that we will be concerned with are {\em Boolean constraint system} games, which will be described in detail in Section \ref{sec:BCS}.

\subsection{Projection games and Parallel repetition }\label{sec:parallelrepetition}
A nonlocal game $\calG=(\calI,\calO,\pi,D)$ is called {\em projection}, if for every $(x,y)\in\supp(\pi)$ and $a\in\calO$, there exists at most one $b\in\calO$ such that $D(x,y,a,b)=1$.

For a positive integer $u$, the $u$ {\em fold repetition} of the nonlocal game $\calG$ is the game $\calG^{\otimes u}:=(\calI^{u},\calO^u,\pi^{\otimes u},D^{\otimes u})$.
Here we denote by $\pi^{\otimes u}$, the probability measure on $\calI^{u}\times\calI^u$, which is defined by $\pi^{\otimes}(x_1,..,x_u,y_1,..,y_u)=\prod_{l\in[u]}\pi(x_l,y_l)$.
The predicate $D^{\otimes u}$ is defined analogously.
We have the following parallel repetition theorem for the class of projection games.

\begin{theorem}[{\cite[Theorem 1.1]{Dinur2015}}]\label{thm:parallelprojection}
    There exist constants $C,c>0$ such that the following holds.
    For any projection game $\calG$,
    $$\omega_q(\calG^{\otimes u})\leq (1-C(1-\omega_q(\calG))^c)^{u/2}.$$
\end{theorem}

We remark that it is possible to replace the above parallel repetition theorem with alternative theorems, such as the one given in \cite{Yuen2016}.
However, the current proof still relies on certain features of projection games.
In particular, these are essential for the definability of item \ref{def:longcodetest}.(5).
It may be possible to generalize the test to accommodate a broader class of BCS games, but we leave this as an open direction for future work.

In addition to the parallel repetition theorem, it has been shown in \cite{Dinur2015} that any nonlocal game can be mapped in polynomial time to a projection game, as we describe next.

Let $\calG=\{\calI,\calO,\pi,D)$ be a given nonlocal game.
Its {\em projection}, denoted by $\calG^{proj}$, is the game in which the referee samples a pair of questions $(i_1,i_2)\sim\pi$, sends both of them to Alice and one of them to Bob, each with probability $1/2$.
Let $i_c$ be the question sent to Bob.
Alice is then required to respond with a pair $(a_1,a_2)\in\calO\times \calO$ and Bob with $b\in\calO$.
They win if and only if $D(i_1,i_2,a_1,a_2)=1$ and $a_c=b$.
We have the following bound on the quantum game value of the projected game.

\begin{lem}[{\cite[Claim 2.8]{Dinur2015}}]\label{lem:proj_by_oggame}
    For any non local game $\calG$,
    $$\omega_q(\calG^{proj})\leq\sqrt{\frac{1+\omega_q(\calG)}{2}}\;.$$
\end{lem}

In general, there is no inequality in the reverse direction. However, we can make the following observation.

\begin{lem}\label{lem:perf_orac_to_perf_proj}
    Suppose that $\calG$ has a perfect oracularizable strategy, then $\calG^{proj}$ also has such a strategy.
\end{lem}
\begin{proof}
    Let $\{A^x_a\}_{a\in\calO}$ be the perfect oracularizable strategy for $\calG$.
    We define the following strategy $p$ for $\calG^{proj}$.
    For every pair $(x,y)\in\supp(\pi)$, we use the PVM $\{A^{xy}_{ab}\}_{(a,b)\in\calO^2}$, which is defined by $A^{xy}_{ab}=A^x_aA^y_b$.
    This is a well-defined PVM.
    For every question $x$, we use the given PVM $\{A^x_a\}_a$.
    Since we assumed that the given strategy is perfect and oracularizable, we have $A^{xy}_{ab}A^x_c=\delta_{ac}A^{xy}_{ab}$ and similarly $A^{xy}_{ab}A^y_c=\delta_{bc}A^{xy}_{ab}$.
    It also implies that if $(x,y)\in\supp(\pi)$ and $a,b\in\calO$ are such that $D(x,y,a,b)=0$, then $1/d\Tr(A^{xy}_{ab})=0$.
    Therefore, $\omega(\calG^{proj},p)=1$.
\end{proof}

In view of the game $\calG^{proj}$, it will be important for us that for any pair of questions that the verifier will sample in $\calG$, there is a pair of answers that satisfies the decider predicate.
To this end, we provide the following lemma.

\begin{lem}\label{lem:no_empty_correct_answers}
    Let $\calG$ be a nonlocal game.
    Then, there exists a nonlocal game $\calG'$, in which to any sampled questions there is a pair of answers that satisfies the decider, and such that $\omega_q(\calG')\leq\omega_q(\calG)+1/2(1-\omega_q(\calG))$
\end{lem}
\begin{proof}
    The question set in $\calG'$ is $\calI^2\cup \calI\cup\calI^2\times\{\pm1\}$ and the answer set is $\calO\uplus\{\pm1\}$.
    Let $P\subset\supp(\pi)$ be the collection of pairs of questions $(i,j)$ such that for every $(a,b)\in\calO^2$, we have $D(i,j,a,b)=0$.
    For every $(i,j)\in \supp(\pi)\backslash P$, we let $\pi'(i,j)=\pi(i,j)$.
    For $(i,j)\in P$, we let $\pi'((i,j,1),(i,j))=\pi'((i,j,-1),(i,j))=\pi(i,j)/2$.
    We set $\pi'$ to be zero on any other pair.
    
    The decider $D'$ is defined as follows.
    Given $(i,j)\in\supp(\pi)\backslash P$, $D'(i,j,a,b)=D(i,j,a,b)$ for every $a,b\in\calO$, and zero otherwise.
    Given $(i,j)\in P$ and $a\in\{\pm1\}$, we set $D'((i,j,a),(i,j),b,c)$ to be $1$ if and only if $a=b=c$.
    
    Let $p=(\{A^i_{a}\},\{B^j_b\},\ket{\psi})$ be a quantum strategy for $\calG'$.
    First, note that given any question of the form $(i,j,a)\in\calI^2\times\{\pm1\}$, we may assume that Alice answers with $a$ (because otherwise they will lose the round for sure).
    That is, $A^{(i,j,a)}_a=\bbI_{\calH_A}$.
    In particular, the combined winning probability in the two rounds that correspond to the question pairs $((i,j,1),(i,j))$ and $((i,j,-1),(i,j))$, is bounded by $$\pi(i,j)(\bra{\psi}A^{(i,j,1)}_1\otimes B^{(i,j)}_1\ket{\psi}+\bra{\psi}A^{(i,j,-1)}_{-1}\otimes B^{(i,j)}_{-1}\ket{\psi})/2= \pi(i,j)/2\;.$$
    Note that since we have $\omega_q(\calG)\leq 1-\pi(P)$, we have $\pi(P)\leq1-\omega_q(\calG)$. Finally,
    \begin{align*}
        \omega(\calG',p) &=\underset{x,y\sim\pi'}{\bbE}\bigg[\sum_{a,b\in\calO'}p(a,b|x,y)D'(x,y,a,b)\bigg]\\
        &\leq\underset{(i,j)\sim\pi}{\bbE}\bigg[1_{(i,j)\notin P}\sum_{a,b\in\calO}p(a,b|i,j)D(i,j,a,b)\bigg]\\
        &+\underset{(i,j)\sim\pi}{\bbE}\bigg[(1_{(i,j)\in P})/2\bigg]\\
        &\leq \omega_q(\calG)+\pi(P)/2\leq \omega_q(\calG)+(1-\omega_q(\calG))/2\;.
    \end{align*}
    
\end{proof}

Note that if $\calG$ admits a perfect strategy, then it must be that $\calG'=\calG$.


\section{BCS and LCS games}
\label{sec:BCS}
\subsection{Boolean constraint system games}
A {\em Boolean constraint} consists of a finite set $S$ of variables (called the {\em context}) and a relation $C\subset \{\pm1\}^S$ (called the {\em constraint}).
The constraint is also considered as a function $C\in\calF_S$.
An assignment $\phi\in\{\pm1\}^S$ {\em satisfies} the constraint $(S,C)$ if $\phi\in C$, or equivalently, $C(\phi)=-1$.
A {\em Boolean constraint system} (BCS) consists of a  (finite) set of variables $X$, and a finite collection of Boolean constraints $\{(S_i,C_i)\}_{i\in[n]}$, where $S_i\subset X$ for every $i\in[n]$.
An assignment $\phi\in\{\pm1\}^X$ {\em satisfies } the BCS $(X,\{(S_i,C_i)\}_{i\in[m]})$, if for every $i\in[m]$, $\phi|_{S_i}$ satisfies $(S_i,C_i)$.

Given a BCS $B=(X,\{(S_i,C_i)\}_{i\in[m]})$ and a probability $\pi$ on $[m]\times[m]$, we can associate a {\em BCS game}, which is the following nonlocal game.
\begin{definition}\label{def:BCSgame}
Denote by $\calG(B,\pi)$ the nonlocal game in which the verifier samples a pair of indices $(i_A,i_B)\sim \pi$, and then send $i_A$ to Alice and $i_B$ to Bob.
The players then need to respond with a pair of assignments $\phi_X\in C_{i_X}$ for $X\in\{A,B\}$.
They win if and only if $(\#)\ \phi_A|_{S_{i_A}\cap S_{i_B}}=\phi_B|_{S_{i_A}\cap S_{i_B}}$.
Formally, $\calG(B,\pi)=([m],\{C_i\}_{i\in[m]},\pi,D)$, where $D$ is defined by the rule $(\#)$.
\end{definition}

In \cite{Dong25}, the authors showed that it is possible to reduce the halting problem to a $\MIP^*$ protocol, with polynomial length questions and answers of constant length.
Culf and Mastel observed in \cite{CM25} that this reduction sends the yes instances to synchronous nonlocal games with perfect oracularizable strategies.
They then use it to build another reduction to a BCS protocol.
We use their reduction in the following construction.

Let $\calG=(\calI,\{\calO_i\},\pi,D)$ be a synchronous nonlocal game.
We will associate $\calG$ with a BCS B and probability measure $\pi^{proj}$ such that $\calG(B,\pi)=\calG^{proj}$.

\begin{definition}\label{def:theProjectedBCS} Let $\calG$ be as above and
suppose that $\calI\subset \{\ 0,1\}^{n}$ and $\calO_i\subset\{0,1\}^{m_i}$, and define the following BCS.
The set of variables $X$ consists of elements $x_{ij}$, for every $i\in \calI$ and $j\in[m_i]$.
We define $S_i=\{x_{ij}\ : \ j\in[m_i]\}$, and we identify the set $\{\pm1\}^{S_i}$ with the set of strings $\{0,1\}^{m_i}$.
Let $C_i\subset \{\pm1\}^{S_i}$ be the image of $\calO_i$ under this identification.
For every pair $(i,i')\in\supp(\pi)$, define the set $S_{ii'}:=S_i\cup S_{i'}$ and $C_{ii'}\subset\{\pm1\}^{S_{ii'}}=\{\pm1\}^{S_i}\times\{\pm1\}^{S_{i'}}$ as the set of pairs $(a,b)$ such that $D(i,i',a,b)=1$.
Note that when $i'=i$, then $S_{ii}=S_i$, so in this case $C_{ii}\subset C_i$ is the subset of elements $a\in\calO_i$ such that $D(i,i,a,a)=1$. 
We let $\calI'=\calI\cup\supp(\pi)$.
Let $\pi^{proj}$ be the probability measure on $\calI'\times \calI'$ that is defined as follows
\begin{align*}
    \pi^{proj}(k,k')\begin{cases}
        \frac{1}{2}\pi(i,i') & \quad k=(i,i'), k'=i \\
        \frac{1}{2}\pi(i,i') & \quad k=(i,i'),k'=i'\\
        0 & \quad \text{otherwise}
    \end{cases}
\end{align*}
\end{definition}


\subsection{Linear constraint system games}\label{sec:LCS}
A {\em linear constraint system} (LCS) $B=(X,\{S_i,C_i\}_{i\in[m]})$ is a BCS in which for every $i\in[m]$, $C_i$ is the collection of satisfying assignments of a Boolean linear equation on the set $S_i$.
Namely, there exists $b_i\in\{\pm1\}$ such that 
$$C_i:=\{\phi\in\bbZ^{S_i}_2\ : \ \underset{x_j\in S_i}{\prod}\phi(x_j)=b_i\}$$

We associate with an LCS the {\em constraint-variable} variation of a constraint system game, which is the following game.
This is in contrast to our convention of associating with a general BCS the {\em constraint-constraint} game, as given in Definition \ref{def:BCSgame}.
For a detailed examination of the relationship between these two game variations for a given BCS and between the corresponding synchronous game values, see \cite{CM25}.

Given an LCS $B$ and a probability $\pi$ on $[m]$, the {\em LCS game} $\calG^{LCS}(B,\pi)$ is defined as follows.
The referee samples $i\sim\pi$ and $x_j\in S_i$ uniformly at random.
She then sends $i$ to Alice and $x_j$ to Bob.
Alice is then required to respond with $\phi\in C_i$, and Bob with $b\in\{\pm1\}$.
They win if and only if $\phi(x_j)=b$.


Next, we present a simple formula for the bias of a synchronous strategy for an LCS game.
To this end, let $B=(X,\{S_i,C_i\}_{i\in[m]})$ be an LCS and $\pi$ a probability on $[m]\times[m]$.
For $i\in[m]$, we denote by $V_i$ the set of $j\in[n]$ such that $x_j\in S_i$.
Let $p$ be a synchronous strategy with PVMs $\{Y^i_{\alpha}\}_{\alpha\in C_i}$ and $\{X^j_{b}\}_{b\in\{\ \pm1\}}$ for all $i\in[m]$ and $1\leq j\leq n$.
Define $A_j^i:=\sum_{a\in C_i}(-1)^{a(j)}Y^i_a$ and $B^j=X^j_1-X^j_{-1}$.
It is an exercise to verify that for all $i\in [m]$, $j'\in V_i$ and $j\in [n]$, $B^j$ and $A^i_{j'}$ are binary observables, and $\{A_j^i\}_{j\in V_i}$ pairwise commute.

For $\alpha\in \{\pm1\}^{S_i}$, define $A^i_{\alpha}:=\prod_{x_j\in \alpha}A^i_j$.
Similarly to Lemma~\ref{lem:fourierprop}, for every $a\in \{\pm1\}^{S_i}$, $Y^i_a=\underset{\alpha\in \{\pm1\}^{S_i}}{\bbE}[(-1)^{a\cdot \alpha}A^i_{\alpha}]$, where $a\cdot \alpha=\sum_{x_j\in S_i}(1+a(x_j)\alpha(x_j))/2$.
Moreover, any strategy for Alice and Bob for an LCS game, in which Alice does not necessarily answer with satisfying assignments, can be turned into a corrected strategy with the same winning probability.
In particular, any synchronous strategy for an LCS game can be described equivalently using binary observables with the above properties.
We have the following.

\begin{lem}\label{lem:observationLCS}
    Let $B=(X,\{S_i,C_i\}_{i\in[m]})$ be an LCS, $\pi$  a probability on $[m]$, and let $p$ be a synchronous strategy with PVMs $\{Y^i_{a}\}_{a\in C_i}$ and $\{X^j_{b}\}_{b\in\{\pm1\}}$.
    Denote by $\{A^i_j\}_{j\in V_i}$ and $B^j$ the corresponding binary observables, as above.
    Then,
    $$\beta(\calG(B,\pi),p)=\underset{i\sim\pi}{\bbE}\big[\underset{j\in V_i}{\bbE}[\langle A^i_j,B^j\rangle]\big]\;.$$
\end{lem}

\begin{proof}
For an equation $i\in[n]$, if $\alpha\in\{\pm1\}^{S_i}$ does not satisfy the $i$th equation, then we set $Y^i_{\alpha}=0$.
Thus,
    \begin{align*}
        2\omega(\calG(B,\pi),p)-1&=2\underset{i\sim\pi}{\bbE}\bigg[\underset{j\in V_i}{\bbE}\bigg[\sum_{\alpha }\langle Y^i_{\alpha},X^j_{\alpha_j}\rangle\bigg]\bigg]-\underset{i\sim\pi}{\bbE}\bigg[\underset{j\in V_i}{\bbE}\bigg[\sum_{\alpha }\sum_{\beta\in\{-1,1\}}\langle Y^i_{\alpha},X^j_{\beta}\rangle\bigg]\bigg] \\
        & = \underset{i\sim\pi}{\bbE}\bigg[\underset{j\in V_i}{\bbE}\bigg[\sum_{\alpha }(\langle Y^i_{\alpha},X^j_{\alpha_j}\rangle-\langle Y^i_{\alpha},X^j_{1-\alpha_j}\rangle)\bigg]\bigg] \\
        & = \underset{i\sim\pi}{\bbE}\bigg[\underset{j\in V_i}{\bbE}\bigg[\sum_{\alpha }(-1)^{\alpha_j}\langle Y^i_{\alpha},B^j\rangle\bigg]\bigg]\\
        & = \underset{i\sim\pi}{\bbE}\bigg[\underset{j\in V_i}{\bbE}\bigg[\langle\sum_{\alpha }(-1)^{\alpha_j} Y^i_{\alpha},B^j\rangle\bigg]\bigg]\\
        & = \underset{i\sim\pi}{\bbE}\big[\underset{j\in V_i}{\bbE}[\langle A^i_j,B^j\rangle]\big]\;.
    \end{align*}
\end{proof}

\section{The long code test for BCS games with entangled provers}\label{sec:test}
For the rest of this section, fix a synchronous nonlocal game $\calG$, and let $B$ and $\pi:=\pi^{proj}$ be as in Definition~\ref{def:BCSgame}. 
We assume that all the constraints in $B$ are nonempty.
This assumption is required for the definability of simultaneously folding over true and conditioning upon a constraint. 
Later, we will ensure that the input BCS for the tester meets this requirement by applying Lemma \ref{lem:no_empty_correct_answers}.


For every finite set $U$ and $f\in\calF_U$, we use the abbreviations $f_U:=s_U(f)$ and $m_f:=m(f\cdot f_U)$.
Given another $C\in\calF_U$, which is not the empty set, recall the definitions of $s_{f,C}$ and $m_{f,C}$.

\begin{definition}[The $L^{\epsilon}(u,B,\pi)$ test]\label{drf:maintest}
    Let $\{B^l=(X^l,\{S^l_i,C^l_i\}_{i\in\calI})\}_{l\in[u]}$ be $u$ distinct copies of B.

    \begin{enumerate}\label{def:longcodetest}
        \item Pick $u$ pairs $\{(i_l',j_l')\}_{l\in[u]}$ according to $\pi$, each pair independently of the others.
        Suppose that $i_l'=(i_l,j_l)$ for every $l\in[u]$, and so $j_l'\in\{i_l,j_l\}$.
        Define $U=\uplus S_{j_l'}^l$, $W=\uplus S_{i_l'}^l$, and $C=\underset{l\in[u]}{\prod}C_{i_l'}^l$.

        \item Choose $f\in \calF_U$ uniformly at random.
        \item Choose $g\in\calF_W$ uniformly at random.
        \item Choose a function $\mu\in\calF_W$ by setting $\mu(y)=1 $ with probability $1-\epsilon$ and $\mu(y)=-1$ otherwise, independently for each $y\in\{\pm1\}^W$.
        \item Set $g'=fg\mu$ to be the function in $\calF_W$ that is defined by $y\mapsto f(y|_U)g(y)\mu(y)$.
        \item Send Alice the tuple $(W,U,C,f,g,g')$.
        \item Choose $X\in\{(U,f_U),(W,s_{g,C}),(W,s_{g',C})\} $ uniformly at random and send it to Bob.
        \item Receive three bits from Alice, $(a_{(U,f_U)},a_{(W,(W,s_{g,C}))},a_{(W,(W,s_{g,C}))})$, and one bit $b\in\{\pm1\}$ from Bob.
        \item Accept if and only if $a_X=b$ and 
        \begin{equation*}\label{eq:code-test}
        a_{(U,f_U)}a_{(W,(W,s_{g,C}))}a_{(W,(W,s_{g',C}))}=m_fm_{g,C}m_{g',C}\;.
        \end{equation*}
    \end{enumerate}

\end{definition}




\begin{prop}[Soundness of the long-code test]\label{prop:soundness}
    Fix some $\epsilon,\delta>0$ and suppose that $\omega_q^s(L^{\epsilon}(u,B,\pi))\geq 1-\frac{1}{36}(1-\delta)^2$.
    Then, $\omega_q(\calG(B,\pi)^{\otimes u})\geq 4\epsilon\delta^2$. 
\end{prop}

\begin{proof}
    First, we explicitly define the sets of queries and variables of the linear system underlying the test.
    Let $\Omega$ be the collection of all possible tuples $(W,U,C,f,g,g')$ that can be generated at the end of step $(5)$ \ref{def:longcodetest}.
    Denote by $R$ the set of all tuples $(U,f_U)$ and $(W,s_{g,C})$ for all possible $U,W,C,g\in\calF_W$ and $f\in \calF_U$ that can be generated in this process.
    We define the set of variables $Z=\{z_t:t\in R\}$.  

    Let $p$ be a synchronous strategy for the test $L:=L^{\epsilon}(u,B,\pi)$ such that $\omega(L,p)\geq1-\frac{1}{36}(1-\delta)^2$.
    For every $\omega=(W,U,C,f,g,g')\in \Omega$, let us denote by $\Tilde{A}^{U}_f, \Tilde{A}^{W,C}_{g}$ and $\Tilde{A}^{W,C}_{g'}$ the binary observables for Alice's answers for equation $\omega$.
    For every $t\in R$, let $B_t$ be the binary observable that corresponds to Bob's answer to the query $t$.

    As in H\r{a}stad original PCP, we take the above observables after conditioning them and folding them over true.
    For every $(W,U,C,f,g,g')\in \Omega$, define $A^U_f:=m_f\Tilde{A}^U_f, A^{W,C}_g:=m_{g,C}\Tilde{A}^{W,C}_g,A^{W,C}_{g'}:=m_{g',C}\Tilde{A}^{W,C}_{g'}$, $B^{U}_f:=m_fB_{(U,f_U)}$, and $B^{W,C}_{g}=m_{g,C}B_{(W,s_{g,C})}$.


    \begin{claim}\label{claim:lowerbound}
        The following holds:
        \begin{align}\label{equ:lowerbound}
            \Big|1-\underset{W,U,C}{\bbE}\big[\underset{f,g,g'}{\bbE}[\frac{1}{d}\Tr(B^U_fB^{W,C}_gB^{W,C}_{g'})]\big]\Big|\leq 1-\delta \;.
        \end{align}
    \end{claim}
    
\begin{proof}
    First, we assume, as we may, that Alice's responses to a query $(W,U,C,f,g,g')\in\Omega$ always satisfy the required equation.
    That is,
    $$A^U_fA_g^{W,C}A^{W,C}_{g'}=\bbI\;.$$
    In particular, the left hand side of Equation \eqref{equ:lowerbound} turns into 
    \begin{equation}\label{eq:dagger}
    \Big|\underset{W,U,C}{\bbE}\big[\underset{f,g,g'}{\bbE}[\frac{1}{d}\Tr(A^U_fA_g^{W,C}A^{W,C}_{g'}- B^U_fB^{W,C}_gB^{W,C}_{g'})]\big]\Big|\;.
    \end{equation}
    To ease notation, let us denote by $Y^{\omega}_j$ the (folded and conditioned) observables of Alice for the query $\omega\in\Omega$, and by $X^{\omega}_j$ the (folded and conditioned) observables of Bob that correspond to the variables in equation $\omega$.
    Using Lemma \ref{lem:general1}, we have the following bound 
    \begin{align*}
        \eqref{eq:dagger} &\leq \underset{\omega\in\Omega}{\bbE}\big[(6(3-\sum_{j=1}^3\langle Y^{\omega}_j,X^{\omega}_j\rangle))^{\frac{1}{2}}\big]\\
        &\leq\big(\underset{\omega\in\Omega}{\bbE}\big[6(3-\sum_{j=1}^3\langle Y^{\omega}_j,X^{\omega}_j\rangle)\big]\big)^{\frac{1}{2}}\\
        &=\big(\underset{\omega\in\Omega}{\bbE}\big[18(1-\frac{1}{3}\sum_{j=1}^3\langle Y^{\omega}_j,X^{\omega}_j\rangle)\big]\big)^{\frac{1}{2}}\\
        & = \big(18(1-\underset{\omega\in\Omega}{\bbE}\big[\frac{1}{3}\sum_{j=1}^3\langle Y^{\omega}_j,X^{\omega}_j\rangle\big])\big)^{\frac{1}{2}}\;,
    \end{align*}
where in the second inequality we used Jensen's inequality $\bbE[X^2]\geq\bbE[X]^2$.
Finally, note that since we conditioned and folded the observables of Alice and Bob that correspond to the same variable in the same way, the last summation is, in fact, the bias of the strategy $p$ in the game.
Thus,
$$|1-\underset{W,U,C}{\bbE}\big[\underset{f,g,g'}{\bbE}[\frac{1}{d}\Tr(B^U_fB^{W,C}_gB^{W,C}_{g'})]\big]|\leq (18(1-\beta(L,p))^{\frac{1}{2}}\;.$$
    The claim follows by the assumption on $\omega(L,p)$.
\end{proof}
    The proof continues almost identically to that in \cite{Hstad2001}. 
    Fix $W,U$ and $C$ as in Definition~\ref{def:longcodetest}(1), and denote by $A_f:=B^U_f$ and $B_g:=B^{W,C}_g$ for $f\in\calF_U$ and $g\in\calF_W$.
    Our focus is on the term   
    \begin{align}\label{equ:fixedWUC1}
        \underset{f,g,g'}{\bbE}[A_fB_gB_{g'}]
    \end{align}
    We replace each observable in \eqref{equ:fixedWUC1} by its Fourier expansion (as we may according to Lemma \ref{lem:fourierprop}(1)):
    \begin{align}
        & \underset{f,g,\mu}{\bbE}\big[\sum_{\alpha,\beta,\beta'}\hat{A}_{\alpha}\hat{B}_{\beta}\hat{B}_{\beta'}\chi_{\alpha}(f)\chi_{\beta}(g)\chi_{\beta'}(g')\big]\nonumber \\
        & = \sum_{\alpha,\beta,\beta'}\hat{A}_{\alpha}\hat{B}_{\beta}\hat{B}_{\beta'}\underset{f,g,\mu}{\bbE}[\chi_{\alpha}(f)\chi_{\beta}(g)\chi_{\beta'}(fg\mu)] \;.\label{equ:fixedWUC2}
    \end{align}
    Following the analysis of $\underset{f,g,\mu}{\bbE}[\chi_{\alpha}(f)\chi_{\beta}(g)\chi_{\beta'}(fg\mu)] $ in Lemma 5.2 in \cite{Hstad2001}, equation \eqref{equ:fixedWUC2} turns into
    \begin{align}
        \sum_{\beta}\hat{A}_{\pi_2(\beta)}\hat{B}^2_{\beta}(1-2\epsilon)^{|\beta|},
    \end{align}
    where we recall that $\pi_2(\beta)\subset\{\pm1\}^U$ contains all the elements $x$ for which there are odd number of elements $y\in\beta$ such that $y|_U=x$.
    So in total, we got so far the following inequality,
    \begin{align}
        \big|\underset{W,U,C}{\bbE}[\sum_{\beta\subset\{\pm1\}^W}\frac{1}{d}\Tr(\hat{B}^U_{\pi_2(\beta)}(\hat{B}^{W,C}_{\beta})^2)(1-2\epsilon)^{|\beta|}]\big|\geq \delta. \label{equ:soundmainlowerbound}
    \end{align}

    Let us describe the strategies of Alice and Bob for the parallel game $\calG(B,\pi)^{\otimes u}$.
    Each player will hold a copy of $\bbC^d$, and they will share a maximally entangled state.
    Suppose that in $\calG(B,\pi)^{\otimes u}$ the verifier sampled the pairs $\{(i'_l,j'_l)\}_{l\in[u]}$.
    Let $W,U$ and $C$ be as in Definition~\ref{def:longcodetest}(1).
    \begin{itemize}
        \item Upon receiving the query $\{j'_l\}_{l\in[u]}$, Bob will measure according to $\{((\hat{B}_{\alpha}^U)^T)^2\}_{\alpha\subset \{\pm1\}^U}$.
        Given that the result of the measurement was $\alpha$, he will choose $\phi\in\alpha$ uniformly at random and will respond with $\phi$.
        \item Upon receiving the query $\{i'_l\}_{l\in[u]}$, Alice will measure according to $\{(\hat{B}^{W,C}_{\beta})^2\}_{\beta\subset\{\pm1\}^W}$.
        Given that the result of the measurement was $\beta$, she will choose $\psi\in\beta$ uniformly at random and will respond with $\psi$.
    \end{itemize}
We claim that the winning probability of Alice and Bob in a given round is bounded by the probability that Alice has measured $\beta$, Bob has measured $\pi_2(\beta)$ times $|\beta|^{-1}$.
First, if Alice's result of the measurement was $\beta$, then by Lemmata \ref{lem:folding} and \ref{lem:conditioning}, $\beta\neq\varnothing$ and for every $\psi\in\beta$, $\psi\in C$.  
If $\phi\in\pi_2(\beta)$, then the last fact implies that for every $l\in[u]$, $ \phi|_{S_{j'_l}}\in C_{j'_l}$.
Finally, for every $\phi\in\pi_2(\beta)$, there exists at least one $\psi\in \beta$ with $\psi|_U=\phi$.
Thus, in total, their winning probability is at least $$\sum_{\beta}\frac{1}{d}\Tr((\hat{B}_{\pi_2(\beta)}^U)^2(\hat{B}^{W,C}_{\beta})^2)|\beta|^{-1}\;.$$
Using the Cauchy-Schwarz inequality (Lemma \ref{lem:CS}), and the fact that the Fourier coefficient of binary observables in the sense of Definition \ref{def:fourier} is Hermitian, we have
\begin{align*}
    & |\sum_{\beta}\frac{(1-2\epsilon)^{|\beta|}}{d}\Tr(\hat{B}_{\pi_2(\beta)}^U(\hat{B}^{W,C}_{\beta})^2)|=|\sum_{\beta}\langle \hat{B}^{W,C}_{\beta}, \hat{B}_{\pi_2(\beta)}^U\hat{B}^{W,C}_{\beta}(1-2\epsilon)^{|\beta|}\rangle| \\
    &\leq (\sum_{\beta}\|\hat{B}^{W,C}_{\beta}\|_{hs}^2)^{\frac{1}{2}}(\sum_{\beta}\|\hat{B}_{\pi_2(\beta)}^U\hat{B}^{W,C}_{\beta}(1-2\epsilon)^{|\beta|}\|_{hs}^2)^{\frac{1}{2}} \\
    &\leq( \sum_{\beta}\frac{(1-2\epsilon)^{2|\beta|}}{d}\Tr((\hat{B}_{\pi_2(\beta)}^U\hat{B}^{W,C}_{\beta})^{\dagger}(\hat{B}_{\pi_2(\beta)}^U\hat{B}^{W,C}_{\beta})))^{\frac{1}{2}}\\
    &=(\sum_{\beta}\frac{1}{d}\Tr((\hat{B}_{\pi_2(\beta)}^U)^2(\hat{B}^{W,C}_{\beta})^2)(1-2\epsilon)^{2|\beta|})^{\frac{1}{2}}\;,
\end{align*}
where we used the fact that $\sum_{\beta}\|\hat{B}^{W,C}_{\beta}\|_{hs}^2=\frac{1}{d}\Tr(\sum_{\beta}(\hat{B}^{W,C}_{\beta})^2)=\frac{1}{d}\Tr(\bbI)=1$.
Using Jensen's inequality $\bbE[X^2]\geq\bbE[X]^2$, we get that
\begin{align}
     \underset{W,U,C}{\bbE}\bigg[\sum_{\beta}\frac{1}{d}\Tr((\hat{B}_{\pi_2(\beta)}^U)^2(\hat{B}^{W,C}_{\beta})^2)(1-2\epsilon)^{2|\beta|})\bigg]& \geq \underset{W,U,C}{\bbE}\bigg[\bigg(\sum_{\beta}\frac{(1-2\epsilon)^{|\beta|}}{d}\Tr(\hat{B}_{\pi_2(\beta)}^U(\hat{B}^{W,C}_{\beta})^2)\bigg)^2\bigg]\nonumber \\
     & \geq \bigg(\underset{W,U,C}{\bbE}\bigg[\sum_{\beta}\frac{(1-2\epsilon)^{|\beta|}}{d}\Tr(\hat{B}_{\pi_2(\beta)}^U(\hat{B}^{W,C}_{\beta})^2)\bigg]\bigg)^2\;. \label{equ:soundalmostfinal}
\end{align}
As observed in \cite{Hstad2001}, we have $|\beta|^{-\frac{1}{2}}\geq (4\epsilon)^{\frac{1}{2}}(1-2\epsilon)^{|\beta|}$.
Combining this with \eqref{equ:soundalmostfinal} and \eqref{equ:soundmainlowerbound} we get the required bound
\begin{align*}
   & \underset{W,U,C}{\bbE}\bigg[\sum_{\beta}\frac{1}{d}\Tr((\hat{B}_{\pi_2(\beta)}^U)^2(\hat{B}^{W,C}_{\beta})^2)|\beta|^{-1})\bigg]\\
   &\geq \underset{W,U,C}{\bbE}\bigg[\sum_{\beta}\frac{1}{d}\Tr((\hat{B}_{\pi_2(\beta)}^U)^2(\hat{B}^{W,C}_{\beta})^2)4\epsilon(1-2\epsilon)^{2|\beta|})\bigg] \\
   & \geq 4\epsilon \bigg(\underset{W,U,C}{\bbE}\bigg[\sum_{\beta}\frac{1}{d}\Tr(\hat{B}^U_{\pi_2(\beta)}(\hat{B}^{W,C}_{\beta})^2)(1-2\epsilon)^{|\beta|}\bigg]\bigg)^2\\ 
& \geq 4\epsilon\delta^2\;.
\end{align*}

\end{proof}

We end this section by showing the completeness of the test.
\begin{lem}\label{lem:test_completeness}
    Suppose that $\calG(B,\pi)$ has a perfect synchronous strategy.
    Then, $$\omega_q^s(L^{\epsilon}(u,B,\pi))\geq1-\epsilon\;.$$
\end{lem}

\begin{proof}
    Let $p$ be a synchronous strategy for $\calG(B,\pi)$ with $\omega(\calG(B,\pi),p)=1$, and let $\{A^{ii'}_{\phi\phi'}\}_{\phi\phi'\in C_{ii'}}$ and $\{A^i_{\phi}\}_{\phi\in C_i}$ be the corresponding PVM's for all $(i,i')\in\supp(\pi)$ and $i\in[m]$.
    Let $\{(i'_l,j'_l)\}_{l\in[u]}$, $W$, $C$,  and $U$ be as in Definition \ref{def:longcodetest}(1).
    We first describe a perfect synchronous strategy for $\calG(B,\pi)^{\otimes u}$, denoted by $p'$.
    In essence, Alice and Bob will execute the strategy $p$ in parallel and independently.
    Denote by $\overline{i}:=(i_1',..,i_u')$ and $\overline{j}:=(j_1',..,j_u')$.
    For every $\phi\in\bbZ_2^W$ and $l\in[u]$ define $\phi_l:=\phi|_{S^l_{i'_l}}$ and $A^{\overline{i}}_\phi:=\bigotimes_{l\in[u]}A^{i_l'}_{\phi_l}$ and for every $\psi\in\bbZ_2^{U}$, define $\psi_l$ similarly, and $A^{\overline{j}}_{\psi}:=\bigotimes_{l\in[u]}A^{j_l'}_{\psi_l}$.
    Note that the above indeed defines valid PVMs for every choice of $\overline{i}$ and $\overline{j}$.
    Also, since $p$ is perfect and the fact that $\Tr(A\otimes B)=\Tr(A)\Tr(B)$,we have $$\frac{1}{d^u}\Tr(A^{\overline{i}}_{\phi}A^{\overline{j}}_{\psi})=\delta_{\phi|_U=\psi}\prod\frac{1}{d}\Tr(A^{i'_l}_{\phi_l}A^{j'_l}_{\psi_l})\;.$$
    Finally, if there exists $i'_l$ such that $\phi_l\notin C_{i'_l}$, then $A^{\overline{i}}_{\phi}=0$, as required.
    Now, it is simple to transform the strategy $p'$ into an almost perfect strategy for $L$.
    In fact, for every $g\in\calF_W$ and $b\in\{\pm1\}$ define $$A^{W,g}_b:=\sum_{\phi:g(\phi)=b}A^{\overline{i}}_{\phi}.$$
    The PVM $A^{U,f}_b$ is defined similarly.  
    Given a tuple $\omega=(W,U,C,f,g,g')$, define the projective measurement with outcomes in $\bbZ_2^3$ by
    $$A^{\omega}_{\overline{b}}:=\sum_{\phi:(f(\phi|_U),g(\phi),g'(\phi))=\overline{b}}A^{\overline{i}}_{\phi}\;,$$
    for every $\overline{b}\in\bbZ_2^3\;.$
    It is immediate to verify that all the above measurements are well defined PVM's.
    The corresponding strategy can be described as follows. 
    Upon receiving the tuple $(W,U,C,f,g,g')$, Alice will execute the measurement $\{A^{\overline{i}}_{\phi}\}$ and respond with $(f(\phi|_{U}),g(\phi),g'(\phi))$.
    Similarly, given a tuple of the form $(U,f)$ or $(W,g)$, Bob will act the same way.
    The corresponding PVMs to Alice and Bob strategy, which we denote by $p''$, are precisely the PVMs $\{A^{\omega}_{\overline{b}}\}$, $\{A^{W,g}_b\}$, and $\{A^{U,f}_b\}$.
    Let us compute the winning probability for this strategy.
    To this end, denote by $b_{\omega}$ the corresponding value on the right-hand side of~\eqref{eq:code-test}.
    In what follows, we use the abbreviation $\phi(f,g,g'):=(f(\phi|_U),g(\phi),g'(\phi))$.
    \begin{align}
        \omega(L,p'')&=\underset{\omega}{\bbE}\bigg[\sum_{\overline{b}:b_1b_2b_3=b_{\omega}}\frac{1}{3}\langle A^{\omega}_{\overline{b}},A^{U,f_U}_{{b_1}}+A^{W,s_{g,C}}_{b_2}+A^{W,s_{g',C}}_{b_3}\rangle\bigg]\nonumber \\
        &=\underset{\omega}{\bbE}\bigg[\sum_{\overline{b}:b_1b_2b_3=b_{\omega}}\frac{1}{3}\sum_{\phi:\phi(f_U,g_W,g'_W)=\overline{b}}\langle A^{\overline{i}}_{\phi},A^{\overline{j}}_{\phi|_U}+A^{\overline{i}}_{\phi}+A^{\overline{i}}_{\phi}\rangle\bigg] \label{eq:a7}\\
        & \geq \underset{\omega'}{\bbE}\bigg[\frac{1}{3}\sum_{\phi}\underset{\mu}{\bbE}[\delta_{\mu(\phi)=1}\langle A^{\overline{i}}_{\phi},A^{\overline{j}}_{\phi|_U}+A^{\overline{i}}_{\phi}+A^{\overline{i}}_{\phi}\rangle]\bigg]\label{eq:a8}\\
        & = \underset{\omega'}{\bbE}\bigg[\frac{1}{3}\sum_{\phi}(1-\epsilon)\langle A^{\overline{i}}_{\phi},A^{\overline{j}}_{\phi|_U}+A^{\overline{i}}_{\phi}+A^{\overline{i}}_{\phi}\rangle\bigg] \label{eq:a9}\\
        & = 1-\epsilon\;.\label{eq:a10}
    \end{align}
    Equality~\eqref{eq:a7} is due to the definition of the measurements and the fact that $p'$ is perfect.
    Therefore, for every $\phi$ such that $A^{\overline{i}}_{\phi}\neq0$, we have $\phi\in C$, thus $(g\wedge C)(\phi)=g(\phi)$ for every such $g\in\calF_W$.
    In particular, we use the notation $g_W$ instead of $s_{g,C}$, to emphasize that it is equal to $g$ or $-g$, depending on the choice made.
    If we fix $\omega'=(W,U,C,f,g)$, and $\phi\in\bbZ_2^{W}$, then for every $\mu$ such that $\mu(\phi)=1$, we have $f_U(\phi|_U)g_W(\phi)(fg\mu)_W(\phi)=b_{\omega}$, for $\omega=(\omega',fg\mu)$.
    This implies~\eqref{eq:a8}.
    Part~\eqref{eq:a9} is due only to the definition of the distribution of $\mu$, and Part~\eqref{eq:a10} is due again to the fact that $p'$ is perfect.
    
\end{proof}

We end this section by further discussing the differences between the reduction presented in \cite{vidick2016three} and the present work. 
As mentioned in the introduction, the main difference is that the earlier paper requires three provers.
The presence of a third prover seems to be helpful in simplifying some of the calculations.
For example, the form of the bias of a three-player strategy for the test made it possible to immediately derive Eq.(10) in \cite{vidick2016three} from the calculations in \cite{Hstad2001}. 
In contrast, in the current work, the transition from the two-player strategy to a sole reliance on Bob's operators required us to establish Lemma \ref{lem:general1} and Claim \ref{claim:lowerbound} to derive equation \eqref{equ:soundmainlowerbound}, which is the counterpart to the step represented by Eq.(10) in the context of H\r{a}stad's proof.
Moreover, at the time of the publication of \cite{vidick2016three} only the inclusion $NP\subseteq \MIP^*$ had been established. Thus, a natural question is whether the more recently established equality $\MIP^*=\RE$ can be used in order to conclude $\RE$-hardness of approximating the quantum value of three-prover protocols with XOR decision predicate. We note that even if such a result was achievable, it would constitute a different form of hardness than the one established in the present work, by virtue of the difference in the range of the reductions. Our main motivation for working with two-prover linear games is the large body of work that places these games at the heart of the connection between quantum interactive proof systems and fundamental questions on the (approximate) representation theory of groups.


\section{Main result}\label{sec:main}
\subsection{MIP* and LIN-MIP*}
Let $1\geq c \geq s\geq 0$ be two constants and let $p,q:\bbR_{\geq 0}\rightarrow\bbR_{\geq0}$ be two functions.
The class $\MIP^*_{c,s}[p,q]$ is the collection of all languages $L\subset\{0,1\}^*$ such that there exist two randomized Turing machines $S$ and $D$ with the following properties.
\begin{enumerate}
    \item For every $x\in\{0,1\}^*$ there exists a game $\calG_x=(\calI_x, \ \{\calO_i^x\}_{i\in\calI_x},\pi_x,D_x)$ such that: \begin{itemize}
        \item We have $\log|\calI_x|\leq p(|x|)$ and $\log|\calO_i^x|\leq q(|x|)$ for every $i\in\calI_x$.
        \item Given the input $x$, machine $S$ runs in time $\poly(|x|)$ and returns a pair $(i,j)\in \calI_x^2$ with probability $\pi_x(i,j)$.
        \item Given the input $x$ and a tuple $(i,j,a,b)\in\calI^x\times\calI^x\times\calO_i^x\times\calO_j^x$,  machine $D$ runs in time $\poly(|x|)$ and returns $D_x(i,j,a,b)$.
    \end{itemize}
    \item If $x\in L$, then $\omega_q(\calG_x)\geq c$.
    \item If $x\notin L$, then $\omega_q(\calG_x)\leq s$.
\end{enumerate}

The class $\LIN^*_{c,s}[p,q]$ is defined similarly except that the games $\calG_x$ are LCS games, and the value $\omega_q$ is replaced by $\omega_q^s$.\footnote{Both $\MIP^*$ and $\LIN^*$ can in principle be defined with larger numbers of rounds and players. For simplicity, we consider the $2$-prover $1$-round case.} We choose to measure the value through the synchronous value $\omega_q^s$, which only considers synchronous strategies, because this leads to a more natural and elegant formulation. Nevertheless, as remarked below, in light of known results that relate the synchronous and quantum values for certain classes of games~\cite{Vidick2022}, this choice only slightly affects concrete constants in the statement of our results. 

We have the following result.
\begin{thm}[{\cite[Theorem 6.7]{Dong25}}]\label{thm:RE_MIP_poly_const} There exist a constant $1>s>0$ and an integer $C$ such that $$\bigcup_{p\in\poly}\MIP^*_{1,s}[p,C]=\RE\;.$$
\end{thm}

Denote by $\halt$ the language of all Turing machines that halt on the empty input.

\begin{corollary}[{\cite{Dong25}},\cite{CM25}]\label{cor:orac_halt}
    We have $\halt\in\MIP^*_{1,s}[\poly,O(1)]$ such that for every $x\in\halt$, the corresponding game $\calG_x$ admits a perfect oracularizable strategy.
\end{corollary}
\begin{proof}
    See details in the proof of Lemma 6.4 in \cite{CM25}.
    
\end{proof}

\begin{thm}\label{thm:main_formal}
For some constant $s'>0$, and for all small enough $\epsilon>0$, we have $\halt\in\LIN^*_{1-\epsilon,s'}[\poly,O(1)]$.
\end{thm}

\begin{remark}\label{rem:quantum_vs_synch}
In the case of a non-halting Turing machine, Theorem~\ref{thm:main_formal} only bounds the synchronous value of the resulting game by $s$. 
It is not necessarily true that the quantum value of a game is attained by its synchronous quantum value. However, in some classes of games, it is possible to relate the quantum game value to the synchronous value. Two such examples are synchronous games and projection games \cite[Theorem 3.1, Theorem 4.6]{Vidick2022}.\footnote{
We note that there is a small mistake in the proof of Theorem 4.6 in \cite{Vidick2022}. 
Indeed, in the notations of the current proof, it is possible that for a given $x$, the sum $\sum_aB^x_a$ is not smaller than $Id$.
Nevertheless, the proof can be easily corrected by switching the roles of Alice and Bob in the definition of the projection game \cite[Definition 4.5]{Vidick2022}.
That is, in the corrected version, for each $(x,y)\in\mathcal{X}\times\mathcal{Y}$, there is $f_{xy}:\mathcal{B}\rightarrow \calA$ such that $D(a,b|x,y)=0$ if $a\neq f_{xy}(b)$.}
As the output game of the reduction implicit in Theorem~\ref{thm:main_formal} is a projection game, the following result is an immediate corollary:
There exists $1>s>0$, such that for every sufficiently small $\epsilon>0$, it is $\RE$-hard to decide if the quantum game value of a given LCS game $B$ is greater than $1-\epsilon$, or smaller than $s$.
\end{remark}

\begin{proof}[Proof of Theorem~\ref{thm:main_formal}]
    Define $\delta:=1-\frac{1}{\sqrt{2}}$ and fix any $1/72>\epsilon>0$.
    We let $s'=1-\frac{1}{36}(1-\delta)^2=71/72$.
    These choices are designed so that $s'<1-\epsilon$.
    For simplicity, we assume $\epsilon$ is a rational number.

    Let $\Tilde{s}$ be as in Theorem \ref{thm:RE_MIP_poly_const}, and denote by $s:=\Tilde{s}+(1-\Tilde{s})/2$.
    Note that if $p\leq q$, then $p+(1-p)/2\leq q+(1-q)/2$.
    Let $C,c>0$ be the constants in Theorem \ref{thm:parallelprojection}.
    Fix $u\in\bbN$ to be such that $$s'':=\bigg(1-C\bigg(1-\sqrt{\frac{1+s}{2}}\bigg)^c\bigg)^{\frac{u}{2}}<4\epsilon\delta^2\;.$$ 
    Let $S,D$ be the randomized Turing machines promised by Theorem \ref{thm:RE_MIP_poly_const}, and let $k:\bbR_{\geq 0}\rightarrow \bbR_{\geq 0}$ be the randomness complexity of $S$.

    First, we claim that we can replace $S$ and $D$ with $S'$ and $D'$ that run in polynomial time, and induce the game $\calG_x'$, as defined in Lemma~\ref{lem:no_empty_correct_answers}.
    In fact, $S'$ use the sampler $S$, and then use $D$ to check if there are winning answers for the sampled questions. 
    This process adds only polynomial time (depending on $D$) to the computations of $S'$, as the set of answers is constant.
    Once receiving an input $x$ and a tuple $(i',j',a',b')$, according to the form of $i'$ and $j'$, $D'$ will validate the correctness of the answer by evaluating $D$ or by direct calculations.
    This will also take only polynomial time.

    Recall the definition of the projection of a nonlocal game given in Section \ref{sec:parallelrepetition}.
    Define the two randomized Turing machines $S^{proj}$ and $D^{proj}$ as follows.
    Given an input $x\in\{0,1\}^*$, and a string $br\in \{0,1\}^{k(|x|)+1}$, $S^{proj}$ outputs the pair $((i,j),j')$ where $S'(x,r)=(i,j)$ and $j'=i$ if $b=0$, and $j'=j$ otherwise.
    Upon receiving an input $x$ and a tuple of the form $((i,j),j',(a,b),c)\in \calI_x^3\times\calO_i^x\times\calO_j^x\times\calO_{j'}^x$, $D^{proj}$ returns 1 if and only if $D_x^{proj}((i,j),j',(a,b),c)=1$, where $D_x^{proj}$ is the decider in the game $\calG_x^{proj}$.
    Both $S^{proj}$ and $D^{proj}$ can be made so that their running time on input $x$ is $\poly(|x|)$.

    Next, recall the definition of the $u$-fold repetition of a nonlocal game, given in Section \ref{sec:parallelrepetition}.
    We define the two randomized Turing machines $S^{proj,u}$ and $D^{proj,u}$ as follows.
    Upon receiving an input $x\in\{0,1\}^*$, $S^{proj,u}$ runs $S^{proj}$ u times independently.
    Clearly, since $u$ is a constant $S^{proj,u}$ runs in $\poly(|x|)$, and it outputs the tuple $\{(i_l',j_l')\}_{l\in[u]}$ with probability $(\pi_x^{proj})^{\otimes u}(\{(i_l',j_l')\}_{l\in[u]})$.
    The decider $D^{proj,u}$ is defined analogously.

    Finally, let us define the randomized Turing machines $S^{\lin}$ and $D^{\lin}$.
    Let $h$ be a constant that upper bounds the length of the answers in the protocols given in Theorem \ref{thm:RE_MIP_poly_const}. 
    Given any element $z\in\{0,1\}^*$, consider the game $\calG_z$.
    For every $i\in\calI_z$, we view $\calO_i^z$ as a subset of $\bbZ_2^h$ first by adding zeros to every element $a\in\calO_i^z$ so that their length is equal to $h$, and then identify $0$ with $1$ and $1$ with $-1$.
    As in the construction given in Definition \ref{def:theProjectedBCS} of the BCS game that corresponds to $\calG^{proj}_{z}$, we denote by $C_i^z\subset\bbZ_2^h$ the image of $\calO_i^z$ under this identification.
    Similarly, we denote by $C_{ii'}^z$ the constraint that corresponds to $(i,i')\in\supp(\pi_z)$, as defined in Definition \ref{def:theProjectedBCS}.
    We fix some enumeration of $\calF_{[h\cdot u]}$ and of $\calF_{[2h\cdot u]}$.
    Also, for every $\overline{1}\neq C\in\calF_{[2hu]}$ we fix a choice between the functions $g\wedge C$ and $(-g)\wedge C$, denoted by $s_{g,C}$, which will be hardcoded in $S^{\lin}$.
    In addition, we fix a distribution $\mu$ over $\calF_{[2hu]}$ such that $\mu(y)=1$ with probability $1-\epsilon$ independently for every $y\in\{\pm1\}^{[2hu]}$.
    We recall that $\epsilon$ is assumed to be rational, so it is possible to construct such a distribution that uses a constant amount of randomness.
    We denote this distribution by $M$.
    
    Let us first describe the sampler $S^\lin$.
    Upon receiving an input $x\in\{0,1\}^*$, the sampler runs first $S^{proj,u}$ to generate a collection $\{(i'_l,j'_l)\}_{l\in[u]}$, where $i_l'=(i_l,j_l)\in \calI_x^2$ and $j'_l\in\{i_l,j_l\}$ for every $l\in[u]$.
    Recall that according to all of our identifications, the set $C:=\prod_{l\in[u]}C_{i_l'}$ is also considered a function of $\calF_{[2hu]}$.
    We also have our well-defined section $s_{[hu]}$.
    In what follows, $S^\lin$ is required to explicitly compute the set $C$.
    For this purpose, it deploys the decider $D^{proj,u}$.
    
    Using additional random coins, $S^\lin$ samples functions $f\in\calF_{[hu]}$, $g\in \calF_{[2hu]}$, and $\mu\sim M$.
    Next, according to the enumeration of $\calF_{[hu]}$ and $\calF_{[2hu]}$, it computes the encoding of $f_{[hu]}:=s_{[hu]}(f)$, $s_{g,C}$ and $s_{fg\mu,C}$, and chooses an element $v$ in $$\{(\{j_l'\}_{l\in[u]},f_{[hu]}),(\{i_l'\}_{l\in[u]},s_{g,C}),(\{i_l'\}_{l\in[u]},s_{fg\mu,C})\}$$ uniformly at random.
    Finally, it outputs the pair $$(*)\ \ \ \big((\{i_l'\}_{l\in[u]}, \{j_l'\}_{l\in[u]},C,f,g,\mu),v\big)\;.$$
    
    Upon receiving an input $x\in\{0,1\}^*$ and a tuple $(\omega,v,\overline{a},b)$ where $(\omega,v)$ is in the form of $(*)$, $\overline{a}\in\bbZ_2^3$ and $b\in\bbZ_2$, $D^\lin$ do the following.
    It first computes the corresponding bit $b_\omega$ as on the right hand side of the equation in \ref{def:longcodetest}(9), and the index $i$ in $[3]$ that corresponds to $v$.
    Then it outputs 1 if and only if $a_1a_2a_3=b_\omega$ and $a_i=b$.

    Observe that $S^\lin$ and $D^\lin$ are valid samplers of an $\MIP^*$ protocol for the language \halt.
    For every $x\in\{0,1\}^*$, the nonlocal game encoded by $S^\lin$ and $D^\lin$, denoted by $\calG^\lin_x$, is the one that corresponds to the test $L^\epsilon(u,B_x,\pi_x^{proj})$, where $B_x$ is the BCS that corresponds to $\calG_x^{proj}$ (where we round the answer sets to have length $h$).
    Since $u$ and $h$ are constant and the elements of $\calI_x$ are of a polynomial length, the set of questions of $\calG^\lin_x$ is also of a polynomial length.
    The length of the answers is bounded by $3$.

    Suppose that $x\in\halt$.
    Then by corollary \ref{cor:orac_halt}, the game $\calG_x$ is oracularizable.
    In particular $\calG_x=\calG_x'$.
    Thus, by lemma \ref{lem:perf_orac_to_perf_proj}, the game $\calG(B_x,\pi_x^{proj})=\calG^{proj}$ admits a perfect synchronous strategy.
    By Lemma \ref{lem:test_completeness}, we have $\omega_q^s(\calG^\lin_x)\geq1-\epsilon$.

    Now, assume that $x\notin\halt$.
    Thus, by the assumption, we have $\omega_q(\calG_x)\leq \Tilde{s}<1$.
    By lemma \ref{lem:no_empty_correct_answers}, $\omega_q(\calG_x')\leq \Tilde{s}+(1-\Tilde{s})/2=s$.
    For readability, we denote $(\calG_x')^{proj}$ simply by $\calG_x^{proj}$.
    By lemma \ref{lem:proj_by_oggame}, $$\omega_q(\calG_x^{prog})\leq\sqrt{\frac{1+s}{2}}<1.$$
    Therefore, by Theorem \ref{thm:parallelprojection}, we have \begin{align*}
        \omega_q^s((\calG^{proj}_x)^{\otimes u})& \leq\omega_q((\calG^{proj}_x)^{\otimes u})\\
        &\leq \bigg(1-C\bigg(1-\omega_q(\calG_x^{proj})\bigg)^c\bigg)^{\frac{u}{2}}  \\  
        &\leq\bigg(1-C\bigg(1-\sqrt{\frac{1+s}{2}}\bigg)^c\bigg)^{\frac{u}{2}}\\
        & = s''<4\epsilon\delta^2\;.
    \end{align*}

    Thus, by Proposition \ref{prop:soundness}, we must have $\omega_q^s(\calG^\lin_x)\leq 1-\frac{1}{36}(1-\delta)^2=s'$.
    That is, $\halt\in\LIN^*_{1-\epsilon,s'}[\poly,3]$, as required.

    We note that the length of the questions in our protocol depends on $\epsilon$, by virtue of the choice of $u$.
    
\end{proof}

\section{Sketch of the analysis for general quantum strategies}\label{sec:general_q.strategies}

The purpose of this section is to outline how to generalize the soundness result, namely, Proposition \ref{prop:soundness}, to the case of general quantum strategies.
Together with Lemma \ref{lem:test_completeness}, it implies the hardness of approximating the quantum value of LCSs.
As explained in remark \ref{rem:quantum_vs_synch}, one can get the same hardness of approximation by applying Theorem 4.6 in \cite{Vidick2022} on Proposition \ref{prop:soundness}.
The advantage of the analysis we sketch here is that it provides a better soundness parameter (namely, $119/120+\delta$, for every sufficiently small $\delta$, instead of $1-(6^{2c}/C)+\delta$ for some two nonnegative integers $C$ and $c$), but, we also believe that it makes the correctness of this result more apparent.

Denote by $\langle\cdot,\cdot\rangle_F$ the Frobenius inner product and by $\|\cdot\|_F$ the induced norm.
Namely, for pair of matrices of finite dimension $d$, $A$ and $B$, we have $\langle A,B\rangle_F=\Tr(A^\dagger B)$ and $\|A\|_F=\langle A,A\rangle_F.$

Let $Y$ and $X$ be operators acting on finite dimensional Hilbert spaces $\calH_A$ and $\calH_B$, respectively, and let $\ket{\psi}\in\calH_A\otimes\calH_B$ be a state, with Schmidt decomposition $\ket{\psi}=\sum_ip_i\ket{u_i}\ket{v_i}$.
Denote by $\lambda:\calH_A\rightarrow\calH_B$ the map $\ket{u_i}\mapsto p_i\ket{v_i}$.
As observed in Lemma 4.1 in \cite{Slofstra2011}, we have $$\|(Y^T\otimes\bbI-\bbI\otimes X)\ket{\psi}\|=\|\lambda Y-X\lambda\|_F.$$

Combining this with the observation that $\lambda\lambda^\dagger=\rho_B$, the reduced density of $\ket{\psi}$ on system $\calH_B$, gives the following generalization of Lemma \ref{lem:general1}.

\begin{lem}\label{lem:inequality_on_three_pairs_of_operators_copy}
    Let $Y_i$ and $X_i$ be binary observables on $\calH_A$ and $\calH_B$, respectively, for $i=1,2,3$, and let $\ket{\psi},\lambda$ and $\rho_B$ be as above.

    Then,
    \begin{align*}
        |\Tr(\lambda Y_1Y_2Y_3\lambda^\dagger-X_1\sqrt{\rho_B}X_2X_3\sqrt{\rho_B})|\leq (18(1-\bbE_i\bra{\psi}Y_i^T\otimes X_i\ket{\psi}))^{\frac{1}{2}}+\|[\sqrt{\rho_B},X_1]\|_F.
    \end{align*}
\end{lem}

Note that Lemma \ref{lem:inequality_on_three_pairs_of_operators_copy} indeed implies Lemma \ref{lem:general1}:
in the latter case we have $\rho_B=\bbI/d$, and so $\|\sqrt{\rho_B},X_1]\|_F=0$.
Its proof is similar to that of Lemma \ref{lem:general1}, and is left for the interested reader.

Next, we wish to generalize claim \ref{claim:lowerbound} to the case of general strategies.
For this purpose, we need the following simple fact, which we record without a proof.
Given an operator $A$ and a scalar $\gamma$, we denote by $\Pi_\gamma(A)$ the orthogonal projection onto the eigenspace of $A$ that corresponds to $\gamma$.

\begin{lem}\label{lem:simple_lemma_copy}
    Let $A$ be a binary observable on a finite dimensional Hilbert space $\calH$, and $\rho$ be a density matrix on $\calH$.
    Then,
    $$\|[A,\sqrt{\rho}]\|_F^2= 4(1-\sum_a\Tr(\Pi_a(A)\sqrt{\rho}\Pi_a(A)\sqrt{\rho})).$$
\end{lem}


\begin{lem}\label{lem:main_tech_copy}
    Let L be an instance of 3LIN, with $n$ equations and $m$ variables, and let $S=A,B,\ket{\psi}$ be a given quantum strategy.
    Denote by $V_i=\{j_1^i,j_2^i,j_3^i\}$ the set of variables that appear in the i'th equation, and let $\lambda$ be as in Lemma \ref{lem:inequality_on_three_pairs_of_operators_copy}.
    Then,
\begin{align*}
    (**)\ &|\bbE_{i}\Tr(\lambda (A^i_{j_3^i}A^i_{j^i_2}A^i_{j^i_1})^T\lambda^\dagger)-\bbE_i\Tr(B_{j^i_1}\sqrt{\rho_B}B_{j^i_2}B_{j^i_3}\sqrt{\rho_B})|\\
    &\leq(18(1-\beta(L,S)))^{\frac{1}{2}}+(12(1-\beta(L,S)))^{\frac{1}{2}}\\
    &\leq (60(1-\beta(L,S)))^{\frac{1}{2}}
\end{align*}
\end{lem}
\begin{proof}[Sketch of proof]
     By Lemma \ref{lem:inequality_on_three_pairs_of_operators_copy}, we have
    \begin{align*}
        (**)&\leq \bbE_i(18(1-\bbE_{j\in V_i}\bra{\psi}A^i_j\otimes B_j\ket{\psi}))^{\frac{1}{2}}+\bbE_i\|[\sqrt{\rho_B},B_{j_i}]\|_F\\
        &\leq(18(1-\bbE_i\bbE_{j\in V_i}\bra{\psi}A^i_j\otimes B_j\ket{\psi})^{\frac{1}{2}}+\bbE_i\|[\sqrt{\rho_B},B_{j_i}]\|_F
    \end{align*}
    where $j_i\in V_i$, for every $i\in[n]$.

    Now, as we assume that Alice always responds with a satisfying assignment, we have the following equalities:
    \begin{align*}
        \omega(L,S)&=\bbE_i\bbE_{j\in V_i}\sum_a\bra{\psi}\Pi_a(A^i_j)\otimes \Pi_a(B_j)\ket{\psi}\\
        &=\bbE_{j\sim \nu}\bbE_{i\sim\nu_j}\sum_a\bra{\psi}\Pi_a(A^i_j)\otimes \Pi_a(B_j)\ket{\psi}\\
        &=\bbE_{j\sim \nu}\sum_a\bra{\psi}\bbE_{i\sim\nu_j}\Pi_a(A^i_j)\otimes \Pi_a(B_j)\ket{\psi},
    \end{align*}
    for some collection of distributions $\nu$ and $\nu_j$ for $j\in[m]$.
    
    Define the POVM with two outcomes $A_j$ by setting $A^a_j=\bbE_{i\sim\nu_j}\Pi_a(A^i_j)$.
    Now, according to Lemma 2.9 in \cite{Vidick2022}, we have  
    \begin{align*}
        \omega(L,S)&\leq\bigg(\bbE_{j\sim\nu}\sum_a\Tr(A^a_j\sqrt{\rho_A}A^a_j\sqrt{\rho_A})\bigg)^{\frac{1}{2}}\bigg(\bbE_{j\sim\nu}\sum_a\Tr(\Pi_a(B_j)\sqrt{\rho_B}\Pi_a(B_j)\sqrt{\rho_B})\bigg)^{\frac{1}{2}}
    \end{align*}

    In particular, as the left factor is smaller than 1, we have
    $$\bbE_{j\sim\nu}\sum_a\Tr(\Pi_a(B_j)\sqrt{\rho_B}\Pi_a(B_j)\sqrt{\rho_B})\geq(\omega(L,S))^2\geq1-2(1-\omega(L,S))=\beta(L,S).$$

    Therefore, using Lemma \ref{lem:simple_lemma_copy}, we have
    \begin{align*}
        \bbE_i\|[\sqrt{\rho_B},B_{j_i}]\|_F&\leq (\bbE_i\|[\sqrt{\rho_B},B_{j_i}]\|^2_F)^{\frac{1}{2}}\leq (3\bbE_i\bbE_{j\in V_i}\|[\sqrt{\rho_B},B_j]\|^2_F)^{\frac{1}{2}} \\
        &=(3\bbE_{j\sim\nu}\|[\sqrt{\rho_B},B_j]\|^2_F)^{\frac{1}{2}}\\
        &=\bigg(12\big(1-\bbE_{j\sim\nu}\sum_a\Tr(\Pi_a(B_j)\sqrt{\rho_B}\Pi_a(B_j)\sqrt{\rho_B})\big)\bigg)^{\frac{1}{2}}\\
        &\leq \bigg(12\big(1-\beta(L,S)\big)\bigg)^{\frac{1}{2}}
    \end{align*}
\end{proof}

\begin{prop}[Soundness of the long-code test]\label{prop:soundness_copy}
    Fix some $\epsilon,\delta>0$ and suppose that $\omega_q(L^{\epsilon}(u,B,\pi))\geq 1-\frac{1}{120}(1-\delta)^2$.
    Then, $\omega_q(\calG(B,\pi)^{\otimes u})\geq 4\epsilon\delta^2$. 
\end{prop}

\begin{proof}[Sketch of proof]
    Let $S=\{A,B,\ket{\psi}\}$ be a quantum strategy for the test $L:=L^{\epsilon}(u,B,\pi)$ such that $\omega(L,S)\geq1-\frac{1}{120}(1-\delta)^2$.
    We use the same notations as in Proposition \ref{prop:soundness}.
    Denote by $\lambda$ the map as in Lemma \ref{lem:inequality_on_three_pairs_of_operators_copy}, $\rho_A$ and $\rho_B$ the reduced densities of $\ket{\psi}$, and by $\ket{\psi_A}$ and $\ket{\psi_B}$ the canonical purifications of $\rho_A$ and $\rho_B$, respectively.

    Lemma \ref{lem:main_tech_copy} implies the following generalization of claim \ref{claim:lowerbound}.
    \begin{claim}\label{claim:lowerbound_copy}
        The following holds:
        \begin{align}\label{equ:lowerbound_copy}
            \Big|1-\underset{W,U,C}{\bbE}\big[\underset{f,g,g'}{\bbE}[\Tr(B^U_f\sqrt{\rho_B}B^{W,C}_gB^{W,C}_{g'}\sqrt{\rho_B})]\big]\Big|\leq 1-\delta \;.
        \end{align}
    \end{claim}

As in \ref{prop:soundness}, line \ref{equ:lowerbound_copy} can be further developed to the following inequality: 
 \begin{align*}
        \big|\underset{W,U,C}{\bbE}[\sum_{\beta\subset\{\pm1\}^W}\Tr(\hat{B}^U_{\pi_2(\beta)}\sqrt{\rho_B}(\hat{B}^{W,C}_{\beta})^2\sqrt{\rho_B})(1-2\epsilon)^{|\beta|}]\big|\geq \delta. \label{equ:soundmainlowerbound}
    \end{align*}

Next we describe the strategy for the u-fold parallel game $\calG(B,\pi)^{\otimes u}$.
In fact, the players will act similarly as the players in the synchronous case, using the measurements $((\hat{B}^U_\alpha)^T)^2$ and $(\hat{B}^{W,C}_\beta)^2$, but each player will hold a copy of $\calH_B$, and they will share the entangled state $\ket{\psi_B}$.
    
One can check that the winning probability of the above strategy is at least
\begin{align*}
    &\sum_{\beta}\bra{\psi_B}(\hat{B}^{W,C}_{\beta})^2\otimes((\hat{B}_{\pi_2(\beta)}^U)^2)^T\ket{\psi_B}|\beta|^{-1}=\\
    &\sum_{\beta}\Tr((\hat{B}_{\pi_2(\beta)}^U)^2\sqrt{\rho_B}(\hat{B}^{W,C}_{\beta})^2\sqrt{\rho_B})|\beta|^{-1}\;.
\end{align*}
Using the Cauchy-Schwarz inequality, we have
\begin{align*}
    & |\sum_{\beta}(1-2\epsilon)^{|\beta|}\Tr(\hat{B}_{\pi_2(\beta)}^U\sqrt{\rho_B}(\hat{B}^{W,C}_{\beta})^2\sqrt{\rho_B})|=|\sum_{\beta}\langle \sqrt{\rho_B}\hat{B}^{W,C}_{\beta}, \hat{B}_{\pi_2(\beta)}^U\sqrt{\rho_B}\hat{B}^{W,C}_{\beta}(1-2\epsilon)^{|\beta|}\rangle_F| \\
    &\leq (\sum_{\beta}\|\sqrt{\rho_B}\hat{B}^{W,C}_{\beta}\|_F^2)^{\frac{1}{2}}(\sum_{\beta}\|\hat{B}_{\pi_2(\beta)}^U\sqrt{\rho_B}\hat{B}^{W,C}_{\beta}(1-2\epsilon)^{|\beta|}\|_F^2)^{\frac{1}{2}} \\
    &\leq( \sum_{\beta}(1-2\epsilon)^{2|\beta|}\Tr((\hat{B}_{\pi_2(\beta)}^U\sqrt{\rho_B}\hat{B}^{W,C}_{\beta})^{\dagger}(\hat{B}_{\pi_2(\beta)}^U\sqrt{\rho_B}\hat{B}^{W,C}_{\beta})))^{\frac{1}{2}}\\
    &=(\sum_{\beta}\Tr((\hat{B}_{\pi_2(\beta)}^U)^2\sqrt{\rho_B}(\hat{B}^{W,C}_{\beta})^2\sqrt{\rho_B})(1-2\epsilon)^{2|\beta|})^{\frac{1}{2}}\;,
\end{align*}
where we used the fact that $\sum_{\beta}\|\sqrt{\rho_B}\hat{B}^{W,C}_{\beta}\|_F^2=\Tr(\rho_B\sum_{\beta}(\hat{B}^{W,C}_{\beta})^2)=\Tr(\rho_B)=1$.
The rest of the proof is identical to that of Proposition \ref{prop:soundness}.

\end{proof}

\printbibliography[title={References}]

\end{document}